\documentclass[12pt]{article}

\usepackage{amsmath,amsthm,amsfonts,braket,algorithm,graphicx}
\usepackage{fullpage,authblk}

\theoremstyle{definition} 
\theoremstyle{definition} 
\newtheorem {theorem} {Theorem}

\newcommand{\kb}[1]{\ket{#1}\bra{#1}}

\newcommand{\bk}[1]{\braket{#1|#1}}

\newcommand{\MR}{\texttt{Measure and Resend}}
\newcommand{\R}{\texttt{Reflect}}

\newcommand{\psqkd}{\Pi^{\texttt{SQKD}}}
\newcommand{\pe}{\Pi^{\texttt{ent}}}
\newcommand{\ps}{\Pi^{\texttt{OW}}}

\newcommand{\rw}{\texttt{\textbf{Rw}}}

\newcommand{\Z}{\mathcal{Z}}
\newcommand{\X}{\mathcal{F}}

\newcommand{\tdl}{\left|\left|}
\newcommand{\tdr}{\right|\right|}

\floatname{algorithm}{Protocol}

\title{High-Dimensional Semi-Quantum Cryptography}
\author[1]{Hasan Iqbal}
\author[1]{Walter O. Krawec\footnote{Email: \texttt{walter.krawec@uconn.edu}}}
\affil[1]{\small{Department of Computer Science and Engineering}\\\small{University of Connecticut}\\\small{Storrs, CT 06269 USA}}

\begin{document}
\maketitle
\begin{abstract}
A semi-quantum key distribution (SQKD) protocol allows two users, one of whom is restricted in their quantum capabilities, to establish a shared secret key, secure against an all-powerful adversary.  In this paper, we design a new SQKD protocol using high-dimensional quantum states and conduct an information theoretic security analysis.  We show that, similar to the fully-quantum key distribution case, high-dimensional systems can increase the noise tolerance in the semi-quantum case.  Along the way, we prove several general security results which are applicable to other SQKD protocols (both high-dimensional ones and standard qubit-based protocols).
\end{abstract}

\section{Introduction}
It is well known that secure key distribution, using only classical communication, is impossible unless computational assumptions are placed on the power of the adversary.  If both $A$ and $B$ are able to communicate using quantum resources, however, perfect security is possible and the only assumption on the adversary required is that she obey the laws of quantum physics.  Quantum Key Distribution (QKD) protocols allow two parties (Alice, $A$, and Bob, $B$) to establish a shared secret key, secure against an all-powerful adversary (Eve, $E$).  Since the first QKD protocol developed by Bennett and Brassard in 1984 (the so-called BB84 protocol \cite{QKD-BB84}), both the theory and practice of QKD has been increasing dramatically.  For a general survey of QKD, both the theory and practice, the reader is referred to \cite{qkd-survey,qkd-survey2,qkd-survey3}.

Since perfect security for key distribution is impossible if both $A$ and $B$ are restricted to classical communication while it \emph{is} possible if both $A$ and $B$ are ``quantum capable,'' a natural question to ask is ``what is the middle ground?''  A communication model designed to help answer this question is the so-called semi-quantum model of cryptography, first introduced in 2007 by Boyer et al. \cite{SQKD-first-PRL}.  In this model, one party is ``fully quantum'' in that they can do anything the protocol requires of them so long as it is possible according to quantum mechanics.  The second party, however, is restricted to operations which are mathematically equivalent to classical communication.  Thus, one party is quantum while the other party is ``classical.''  Since its original introduction in 2007, there have been numerous new semi-quantum key distribution (SQKD) protocols developed \cite{SQKD-second,SQKD-mirror,SQKD-lessthan4,SQKD-no-measure1,med-sqkd1,med-sqkd2,med-sqkd3}.  There have also been extensions to the model beyond basic key distribution including secret sharing \cite{SQKD-secret1,SQKD-secret2,SQKD-secret3} and state comparison \cite{state-comp1,state-comp2,state-comp3}.

(S)QKD protocols operate in two stages: first is a \emph{quantum communication stage} whereby users utilize the quantum channel, along with the classical authenticated channel (which is a classical communication channel that is authenticated, but not secret), to establish a \emph{raw key}.  $A$ and $B$ both have their own raw key which  is a string of classical bits that are partially correlated (there may be some errors due to an adversary's attack or just natural noise) and partially secret (an adversary may have some information on this raw key).  Thus this raw-key by itself cannot be used directly as a secret key.  Users, therefore, must run a second stage where, at a minimum, they will execute an error correction protocol using the authenticated classical channel (leaking additional information to $E$) followed by a privacy amplification protocol which takes the error-corrected raw key and hashes it down to a smaller secret key.  The relative size of the secret key compared to the initial raw-key (called the \emph{key-rate} of the protocol) is a statistic of great importance in QKD research and bounding it, as a function of observed noise in the quantum channel, is the main challenge in any (S)QKD security proof.  A related statistic is the \emph{noise tolerance} of the protocol which specifies the noise threshold after which the adversary has too much information and so users must simply abort (e.g., BB84's noise tolerance is $11\%$ \cite{QKD-BB84-rate1,renner-keyrate}).  Before this tolerance threshold is reached, privacy amplification is able to give a positive, though potentially small, key (the size of the secret key decreases as the noise increases due to the direct correlation between noise and adversarial information gain).

As far as semi-quantum cryptography is concerned, there have been, by now, several proofs of security based on key-rate computations for SQKD protocols and, rather surprisingly, despite the limitations on one of the users, along with the increased attack strategy space afforded to the adversary (due to the requirement of a two-way quantum channel allowing quantum resources to travel from $A$, to $B$, then back to $A$), noise tolerances compare favorably to several fully quantum protocols.  In particular, in \cite{QKD-Tom-Krawec-Arbitrary}, the noise tolerance of the original Boyer et al., protocol can approach $11\%$, the same as BB84.  However, this optimistic result required looking at numerous \emph{mismatched statistics} (a technique introduced in \cite{QKD-Tom-First}, extended for one-way channels in \cite{QKD-Tom-KeyRateIncrease,QKD-Tom-KeyRateMismatchedDistill,QKD-Tom-BB84NarrowAngle}, and expanded for two-way quantum channels in \cite{QKD-Tom-Krawec-Arbitrary}).  Without these statistics, and only looking at the error rate, the Boyer et al., protocol has a noise tolerance of $6.14\%$, though this is only a lower-bound and future refinements to the security proof techniques may improve this to the $11\%$ found with mismatched measurements.  Currently, the best-known noise tolerance for an SQKD protocol is from \cite{SQKD-high-noise} which can attain a tolerance of $17.8\%$ or even as high as $26\%$ for certain, practical, quantum channels (a result comparable to BB84 with Classical Advantage Distillation \cite{CAD-1,BB84-CAD}).  Again, this high tolerance bound required looking at numerous  mismatched measurements.

Designing protocols with increased noise tolerance is an important task.  Encouraged by recent theoretical successes in fully-quantum QKD using high-dimensional carriers \cite{high-dim0,high-dim1,QKD-high-dim-dep,QKD-high-dim-50,QKD-QW,high-dim2,high-dim3,high-dim4,high-dim5,high-dim6,high-dim7} and, in particular, these protocols' ability to withstand high channel noise levels (some reaching $50\%$ as the dimension of the quantum carrier approaches infinity), we ask, can a high dimensional quantum communication channel also benefit semi-quantum key distribution?  Or does this substantial improvement in noise tolerance require two fully quantum users to truly harness?  We note that a high-dimensional SQKD protocol was introduced in \cite{QKD-QW}, using a quantum walk as the information carrier, however a noise tolerance computation was not performed due to the great complexity of that protocol and so this question still remained open (\emph{though the methods we develop in this paper may be applicable to other protocols such as this quantum-walk based protocol}).

In this paper, we show high dimensional states can benefit semi-quantum communication and in doing so, make several contributions in this work.  We design a new high-dimensional SQKD protocol and conduct an information theoretic security analysis allowing us to compute a lower-bound on its key-rate based on observed channel noise.  Our security proof introduces several new techniques which may be applicable to other (S)QKD protocols (both standard qubit-based and future-developed high-dimensional ones including, perhaps, the quantum-walk SQKD protocol developed in \cite{QKD-QW}).  Semi-quantum protocols rely on a two-way quantum channel giving the adversary a greater attack strategy space making security analyses for semi-quantum protocols difficult, especially in higher dimensions (all past work involving key-rate computations have been for the qubit case). As such, our new methods may prove beneficial not only for other semi-quantum protocols, but also fully-quantum protocols reliant on a two way quantum channel (of which there are several \cite{lm05,pp-sdc,lm05-ent,CV-floodlight,CV4,CV5}).  Finally, we evaluate our protocol's performance and determine its noise tolerance for varying dimensions and show that, indeed, high dimensional carriers do benefit the noise tolerance of semi-quantum protocols.  We show that our protocol's noise tolerance tends to $30\%$ as the dimension increases; this result is without requiring any mismatched statistics.  While this is not as high as the $50\%$ achieved in the fully quantum case, this is still higher than any other SQKD protocol to-date and, considering that this is a semi-quantum protocol, where one participant is severely limited in their capabilities, is still a very positive result.  This work paves the way for future research in higher-dimensional systems for semi-quantum or two-way quantum cryptography.  By analyzing semi-quantum protocols with high dimensional systems, we further map out the ``gap'' between classical and quantum communication systems.

In this work, we are primarily concerned with a theoretical protocol and not practical attacks or complications involving its implementation.  We note that several fully quantum high-dimensional QKD protocols have been experimentally implemented and the experimental generation of high dimensional entangled states has seen rapid progress lately \cite{high-dim5-exp,high-dim5-exp2,high-dim-exp3,high-dim-exp4}.  However, we do not concern ourselves with practical implementations of this system. Instead, we are solely interested in understanding how high-dimensional quantum states may benefit the semi-quantum model of cryptography, leaving practical issues as future work.

\section{Preliminaries}

If $\rho_{AB}$ is a \emph{density operator} (i.e., a Hermitian positive semi-definite operator of unit trace) acting on Hilbert space $\mathcal{H}_A\otimes\mathcal{H}_B$, then we write $\rho_A$ to mean the partial trace over the $B$ portion, namely $\rho_A = tr_B\rho_{AB}$.  Similarly for other, or multiple, systems.  Given a system $\rho_{AB}$ which is unmeasured, and an orthonormal basis $V = \{\ket{v_1},\cdots,\ket{v_d}\}$ for the $A$ system (which is of dimension $d$), then we write $\rho_{A^VB}$ to mean the density operator resulting from a measurement of the $A$ register in this $V$ basis.

We use $H(A)_{\rho}$ to mean the entropy function - either the classical Shannon entropy (if $\rho$ is a classical state) or the quantum von Neumann entropy (the context will always be clear which we mean).  Note this implies $\rho$ is a density operator acting on \emph{at least} some $A$ register (if it acts on others, we first trace out those additional spaces and compute the entropy in the resulting $A$ space only).  The von Neumann entropy is defined: $H(A)_\rho = H(\rho_A) = -tr(\rho_A\log\rho_A)$ (where all logarithms in this paper are base two).  The conditional entropy is denoted $H(A|B)_{\rho}$ and defined $H(AB)_\rho - H(A)_\rho$.  If $\rho_{AB}$ is an unmeasured quantum state, then by $H(A^V|B)_{\rho}$ we mean the conditional entropy in the operator resulting from measuring the $A$ portion of $\rho_{AB}$ in the $V$ basis (the $B$ portion remains unmeasured).  By $H(A^V|B^V)_{\rho}$ we  mean the same, but after also measuring the $B$ portion (in which case the entire state is classical and so Shannon entropy is used).  If the context is clear, we may drop the subscript.  Finally, for a real number $x\in[0,1]$ we write $H(x)$ to mean the binary Shannon entropy, namely $H(x) = -x\log x - (1-x)\log(1-x)$.

Given operator $X$, we write $||X||$ to mean the trace distance.  If $X$ is Hermitian and finite dimensional, then this is simply the sum of the absolute values of the eigenvalues of $X$.

Finally, let $\rho_{ABE}$ be a quantum state where the $A$ portion is $d$-dimensional and let $V = \{\ket{v_1},\cdots,\ket{v_d}\}$ and $U=\{\ket{u_1},\cdots,\ket{u_d}\}$ be two orthonormal bases.  An important \emph{entropic uncertainty relation} which will be used later, was proven in \cite{QKD-uncertainty} and states that for any quantum state $\rho_{ABE}$, it holds that:
\begin{equation}\label{eq:uncertainty}
H(A^V|E)_\rho + H(A^U|B) \ge -\log c,
\end{equation}
where $c = \max_{i,j}|\braket{v_i|u_j}|^2$.  This will be used in our proof of security later.

\subsection{Semi-Quantum Cryptography}
The semi-quantum model, as introduced in \cite{SQKD-first-PRL}, consists of at least one ``fully-quantum'' user (typically $A$) and one ``classical'' or ``semi-quantum'' user (typically $B$).  This classical user is only allowed to interact with the quantum channel in a very restricted way.  In particular, he can choose to do one of two things on receiving any quantum state from $A$:
\begin{itemize}
  \item $\R$: If he chooses this option, he will disconnect from the quantum channel, creating a loop back to $A$.  In this case, the quantum user is simply ``talking to herself'' over a large, looped, quantum channel.
  \item $\MR$: If he chooses this option, he will perform a measurement of the quantum state in a single, publicly known, basis (typically the computational basis).  Based on his measurement result, he will then send a new quantum state, prepared in this same basis, back to $A$.
\end{itemize}

Clearly, if both users are semi-quantum and can only perform these two operations, the system is mathematically equivalent to a classical communication protocol as both users would be restricted to only operating directly in a single, publicly known, basis.  Thus, the interest in semi-quantum cryptography is to see how security holds when one user is quantum, but the other is classical according to the above functionality.

Note that we are not considering practical device security in this work and are only interested in the theoretical properties of semi-quantum communication.  Thus we do not concern ourselves with such attacks as the photon tagging attack \cite{SQKD-photon-tag,SQKD-photon-tag-comment} or multi-photon attacks (especially problematic when $B$ chooses $\MR$ as he must re-prepare qubits in the observed state).  Though interesting, these are outside the scope of this work - techniques from \cite{SQKD-mirror} may prove beneficial to securing our protocol against these attacks but we leave this as interesting future work.

As mentioned earlier, (S)QKD protocols operate, first, through a quantum communication stage.  This stage utilizes the quantum communication channel and the authenticated classical channel to output a \emph{raw-key} of size $N$ bits.  From this error correction and privacy amplification are run outputting a secret key of size $\ell(N)$ bits.  The \emph{key-rate} is defined to be the ratio $\ell(N)/N$.  We are interested in the theoretical asymptotic limit.  In this case, assuming collective attacks (i.i.d. attacks where $E$ is free to store a quantum memory system for measurement at any future point in time \cite{qkd-survey}), it was shown in \cite{renner-keyrate,winter-keyrate} that:
\begin{equation}
r = \lim_{N\rightarrow \infty}\frac{\ell(N)}{N} = \inf(H(A|E)_\rho - H(A|B)_\rho),
\end{equation}
where $\rho_{ABE}$ is a density operator describing a single iteration of the quantum communication stage, conditioned on that iteration being used to distill raw key material (i.e., not on an iteration used only for error checking or an iteration that is later discarded due to an incompatible basis choice).  The infimum is over all collective attacks that induce the observed noise statistics.  Above, the $A$ and $B$ registers are the actual classical raw-key bit registers and only the $E$ portion is quantum.  It is this entropy equation, and in particular the von Neumann entropy $H(A|E)$, that we are interested in computing and is the main challenge (computing $H(A|B)$ is generally trivial given the observed noise statistics).

Our protocol uses higher-dimensional systems and, as such, we must define the bases we work with.  For the classical user, we will use the computational basis of dimension $2^n$, namely: $\{\ket{0\cdots 00}, \ket{0\cdots01}, \cdots, \ket{1\cdots 11}\}$ which, when needed to simplify  notation, we will also label equivalently as $\{\ket{0}, \ket{1}, \cdots, \ket{2^n-1}\}$.  We use $\Z$ to denote this basis.

The quantum user, of course, is not restricted to operating in only one basis and so we also define the following ``$\X$'' basis:
\begin{equation}
\X = \{\ket{F_0}, \ket{F_1}, \cdots, \ket{F_{2^n-1}}\},
\end{equation}
where $\ket{F_x} = \mathcal{F}\ket{x}$ and $\mathcal{F}$ is the quantum Fourier transform, namely:
\begin{equation}
\mathcal{F}\ket{x} = \frac{1}{\sqrt{2^n}}\sum_{y=0}^{2^n-1}\exp(-\pi i  x y/2^{n-1})\ket{y}.
\end{equation}

Of course, one may consider other bases that the quantum user may utilize.  However, our protocol will make use of both the $\Z$ and $\X$ bases.  Note that, for the classical user, if he chooses $\MR$ or $\R$, that operation is performed on an entire $n$-qubit signal state (e.g., he cannot reflect ``half'' the qubits and measure the other half in our model).

\section{The Protocol}
Our protocol is shown in Protocol \ref{prot:sqkd}.  The protocol operates by having $A$ send signals of $n$-qubits each.  For each iteration, $B$ will either $\MR$ the entire $n$-qubit state or he will $\R$ the entire state.  Whenever $A$ sends a $\Z$ basis state and $B$ chooses to $\MR$, they will add $n$ bits to their raw key.  Once a sufficiently large raw key has been established, standard error correction and privacy amplification are run.  In the next section we will compute a lower-bound on the key-rate of this protocol.  We will consider a noisy, but loss-less, quantum channel and ideal devices.  Practical security concerns, though interesting, are outside the scope of this work and would provide interesting future work.  Any collective attack against this protocol consists of two unitary operators $(U_F, U_R)$ where $U_F$ is applied in the forward channel and $U_R$ is applied in the reverse.

\begin{algorithm}
\caption{$n$-dimensional SQKD: $\psqkd$}\label{prot:sqkd}
\textbf{Public Parameters:} $n$: the number of qubits to send per signal; $p_M$, the probability of choosing $\MR$; $p_Z$, the probability of $A$ choosing the $\Z$ basis.

$ $\newline
\textbf{Quantum Communication Stage:} The quantum communication stage of the protocol will repeat the following until a sufficiently large raw-key has been distilled:
\begin{enumerate}
  \item With probability $p_Z$, $A$ prepares a randomly chosen $\Z$ basis state; otherwise she prepares a randomly chosen $\X$ basis state.  She records her choice of basis and the choice of state and sends the resulting $n$-qubit state to $B$.

\item $B$ chooses, with probability $p_M$ to $\MR$, measuring all $n$ qubits in the computational basis and recording the result and then resending the observed state back to $A$.  Otherwise, with probability $1-p_M$, he chooses $\R$ in which case he reflects all $n$ qubits back to $A$.

\item $A$ measures the returning $n$ qubit system in the same basis she used to prepare.

\item $A$ and $B$, using the authenticated classical channel, divulge their choices ($B$ his choice of ``$\MR$'' or ``$\R$'' and $A$ her choice of basis).  If $A$ chose the $\Z$ basis and $B$ chose $\MR$, they will use this iteration to contribute towards their raw key; namely, $B$ will append his $n$-bit measurement result string and $A$ will append her initial state she prepared to their respective raw-keys (in this case, $A$'s subsequent measurement result is not used).  We call this a \emph{key-distillation iteration}.  Otherwise, this iteration (along with a suitably chosen random subset of key-distillation iterations) may be used for error detection in the obvious way.
\end{enumerate}
\end{algorithm}

\section{Security Analysis}
We now analyze the security of our protocol.  As with other (S)QKD protocols, we show security against collective attacks.  We will comment on general attacks later.  Our security analysis extends ideas introduced in our conference paper \cite{krawec-conf} but to the higher-dimensional case and consists of two main parts: First, we will prove that it is sufficient to analyze a particular \emph{one-way} fully quantum protocol and, once security is proven there assuming the same channel observations are made, security of our SQKD protocol follows immediately.  \emph{This reduction is very general and can apply to other SQKD protocols}.  Thus, to analyze security of the two-way semi-quantum protocol, it suffices to consider a particular one-way protocol which is easier to analyze as $E$ only attacks once.  Second, we analyze the security of this one-way protocol through the use of entropic uncertainty relations, and continuity of conditional von Neumann entropy.  The techniques we develop in both steps are often general and may be applicable to other two-way (S)QKD protocols.

\subsection{Reduction to a One-Way Protocol}
In this section we show how certain SQKD protocols, of arbitrary dimensions, may be reduced to a one-way protocol.  Note that in \cite{lm05-ent}, a method of reducing two-way fully quantum protocols to one-way, entanglement based protocols was shown, however that method only applies if the original protocol admits a certain symmetry property which semi-quantum protocols necessarily lack (due to $B$'s use of $\MR$).  As a first step, we first consider an intermediate, two-way, SQKD protocol, which we denote by $\pe$.  This intermediate protocol is no longer prepare-and-measure, but instead has $A$ preparing entangled qudits and $B$ performing a CNOT gate whenever he chooses $\MR$.  The protocol is shown in Protocol \ref{prot:ent}.  It is not difficult to see that security of $\pe$ implies security of $\psqkd$ (i.e., $\pe \Rightarrow \psqkd$ where ``$\Rightarrow$'' means ``implies security of'').  Indeed, $A$'s prepare-and-measure scheme in $\psqkd$ is equivalent to her preparing the entangled state of $2n$ qubits $\ket{\psi_0} = \frac{1}{2^{n/2}}\sum_{a}\ket{a,a}$ and sending the right register (consisting of $n$ qubits) to $B$ while keeping the left-half to herself.  If $B$ chooses to reflect, this is nothing more than an identity operation whereas if he chooses to $\MR$, then by applying CNOT gates targeting his register and then measuring at some future time, this is equivalent to him measuring immediately.  Finally, when qubits return to $A$, she may measure both $n$ qubit registers in the same basis - standard arguments \cite{qkd-survey,SQKD-second} show that her measurement of the $A_1$ register is equivalent to her initially preparing the state she observes at this later point.  Furthermore, a collective attack against this protocol is identical to the $\psqkd$ case, namely two unitary attack operators $(U_F, U_R)$.

\begin{algorithm}
\caption{Entanglement-Based $n$-dimensional SQKD: $\pe$}\label{prot:ent}
\textbf{Public Parameters:} $n$: the number of qubits to send per signal; $p_M$, the probability of choosing $\MR$; $p_Z$, the probability of $A$ measuring in the $\Z$ basis.

$ $\newline
\textbf{Quantum Communication Stage:} The quantum communication stage of the protocol will repeat the following until a sufficiently large raw-key has been distilled:
\begin{enumerate}
  \item $A$ prepares the $2n$-qubit state: $\ket{\psi_0} = \frac{1}{\sqrt{2^n}}\sum_{a=0}^{2^n-1}\ket{a,a}_{A_1T}$
and sends the ``$T$'' portion to $B$.

\item $B$ chooses, with probability $p_M$ to $\MR$ in which case he applies the operator $CNOT^{\otimes n}$, acting on the $T$ space and his own private $B$ register (also of $n$ qubits).  Otherwise, with probability $1-p_M$, he chooses $\R$ and applies $I^{\otimes n}$ to the $T$ portion (thus, his $B$ register will remain independent of the system in this case).  Either way, the $T$ register is then returned to $A$.  Once returned, we rename the $T$ register as the $A_2$ register. 


\item $A$ chooses to measure in the $\Z$ basis (with probability $p_Z$) or the $\X$ basis (with probability $1-p_Z$).  She measures both the $A_1$ register and the returned $T$ register (now called the $A_2$ register) in the same basis (either both $\Z$ or both $\X$).  At this point, $B$ will measure his register in the $\Z$ basis if he chose $\MR$.

\item $A$ and $B$ divulge their choices ($B$ his choice of ``$\MR$'' or ``$\R$'' and $A$ her choice of basis).  If $A$ choose the $\Z$ basis and $B$ chose $\MR$, they will save their measurement results and append the resulting value (as a bit-string) to their respective raw-keys ($A$ will use her result from the $A_1$ register, discarding the $A_2$ register in this case).
\end{enumerate}
\end{algorithm}

Next, we introduce our one-way protocol, shown in Protocol \ref{prot:one-way} and denoted $\ps$.  At first glance, the two protocols, $\pe$ (which is semi-quantum and uses a two-way quantum channel) and $\ps$ (which is one-way and fully quantum) do not appear similar.  However, we will prove that security of $\ps$ implies security of $\pe$ (which, in turn, implies security of our actual protocol $\psqkd$).  We do this by showing that, for any attack against $\pe$, there exists an attack against $\ps$ which causes $E$ to gain as much information on the raw-key as in $\pe$ and, furthermore, the view according to $A$, $B$, and $E$ are identical in both cases (i.e., the two cases are indistinguishable).  Thus, if we analyze $\ps$ (which is easier to do since it is one-way), we automatically cover any attack against $\pe$.  Ultimately, this technique is an extension of a result in our conference paper \cite{krawec-conf} to the arbitrary, $N$-dimensional case (only the qubit, $N=2$ case was considered before).  However, beyond being more general, our proof here is also more refined as it does not require an additional ``simplification'' step that was necessary in \cite{krawec-conf}.

Let $N = 2^n$.  An attack against $\ps$ consists of a probability distribution $\{p(b)\}$ for all $b=0,1,\cdots, N-1$ along with a \emph{single} attack operator $U$ acting on $2n$ qubits and $E$'s quantum ancilla.  Note that $E$ gets to choose the values $p(b)$ which $B$ uses to prepare his states - thus, $E$ has partial control over $B$'s source device in $\ps$; the reason for this necessity will be apparent later in our proof.  We now prove it is sufficient to consider security of $\ps$ (in which case we have $\ps \Rightarrow \pe \Rightarrow \psqkd$).

\begin{algorithm}
\caption{One-Way $n$-dimensional QKD: $\ps$}\label{prot:one-way}
\textbf{Public Parameters:} $n$: the number of qubits to send per signal; $p_M$, the probability of choosing $\MR$; $p_Z$, the probability of $A$ measuring in the $\Z$ basis; $\{p(b)\}_{b=0}^{2^n-1}$, probability values set by the adversary but known to all parties.

$ $\newline
\textbf{Quantum Communication Stage:} The quantum communication stage of the protocol will repeat the following until a sufficiently large raw-key has been distilled:
\begin{enumerate}
\item $B$ chooses, with probability $p_M$ operation ``$\MR$'' otherwise he chooses ``$\R$.''  Note that the terminology $\MR$ and $\R$ do not have any operational meaning in this protocol - we simply use them so that the reduction later from our SQKD protocol $\psqkd$ makes sense.  If he chooses $\R$, he prepares a $3n$ qubit state of the form:
\begin{equation}
\ket{\phi_R} = \sum_{b=0}^{2^n-1}\sqrt{p(b)}\ket{b,b}_{A_1A_2} \otimes \ket{0}_B
\end{equation}
where the right-most $B$ register contains $n$ qubits in the state $\ket{0}$.  Otherwise, if he chooses $\MR$, he prepares a $3n$ qubit state of the form:
\begin{equation}
\ket{\phi_{MR}} = \sum_{b=0}^{2^n-1}\sqrt{p(b)}\ket{b,b,b}_{A_1A_2B}
\end{equation}
Regardless of his choice, he sends the $A_1A_2$ register (consisting of $2n$ qubits) to $A$.

\item Same as step (3) of $\pe$.

\item Same as step (4) of $\pe$.
\end{enumerate}
\end{algorithm}

\begin{theorem}\label{thm:security1}
Let $(U_F, U_R)$ be a collective attack against $\pe$ and let $\rho_{ABE}$ be the resulting density operator describing a single iteration of $\pe$ in the event this attack is used.  Then, there exists an attack of the form $(\{p(b)\}_{b=0}^{2^n-1}, U)$ against $\ps$ such that, if $\sigma_{ABE}$ is the resulting density operator of a single iteration of $\ps$ in this case, it holds that $\sigma_{ABE} = \rho_{ABE}$.  In particular, there is no advantage to $E$ in either case and, furthermore, no party $A$, $B$, or $E$ can distinguish between the two scenarios.
\end{theorem}
\begin{proof}
Fix an attack $(U_F, U_R)$. Without loss of generality, we may write $U_F$'s action on basis states as:
\[
U_F\ket{a}\otimes\ket{\chi}_E = \sum_{b=0}^{N-1}\ket{b, e_{a,b}},
\]
where $N = 2^n$ and $\ket{e_{a,b}}$ are arbitrary states in $E$'s ancilla (we assume, without loss of generality in the collective attack case, that $E$'s ancilla starts in some pure state $\ket{\chi}_E$).  Unitarity, of course, imposes some restrictions on these states.  In particular, for every $a$ it holds that:
\begin{equation}\label{eq:thm1-norm}
\sum_{b=0}^{N-1}\bk{e_{a,b}} = 1.
\end{equation}

Given this attack, we construct $(\{p(b)\}, U)$, an attack against $\ps$, that satisfies the theorem statement.  To do so, we follow a technique first introduced in our conference paper \cite{krawec-conf} but generalized here for higher dimensions.  First, we set the values $p(b)$ to:
\begin{equation}
p(b) = \frac{1}{N}\sum_{a=0}^{N-1}\bk{e_{a,b}}
\end{equation}
Clearly $p(b) \ge 0$ for all $b$.  Furthermore, from Equation \ref{eq:thm1-norm}, it follows that:
\[
\sum_bp(b) = \frac{1}{N}\sum_b\sum_a\bk{e_{a,b}} = \frac{1}{N}\sum_a\sum_b\bk{e_{a,b}} = 1,
\]
thus this is a valid probability distribution, and so a valid attack setting.

Now, consider the following operator $\rw$ which we call the ``rewind'' operator as, in a way, it ``rewinds'' the channel so that a state prepared by $B$ in the one-way case (i.e., protocol $\ps$) appears \emph{to all three parties} as if it had been prepared by $A$ in the two-way case (i.e., $\pe$).  In particular, it will ``setup'' the $A_1$ register and $E$'s quantum memory as if this had been performed in the two-way $\pe$ case.  The only thing that cannot be ``rewound'' is $B$'s measurement distribution, thus the need for $E$ to set this separately through the $p(b)$ values.  This operator acts on basis states $\ket{b,b}$ (sent by $B$ in the one-way protocol $\ps$) as follows:
\begin{equation}
\rw\ket{b,b}_{A_1A_2} = \frac{\sum_{a=0}^{N-1}\ket{a,b,e_{a,b}}}{\sqrt{N\cdot p(b)}}.
\end{equation}
It is not difficult to see that $\rw$ is an isometry.  Indeed, given $\ket{b,b}$ and $\ket{b',b'}$ for $b \ne b'$, we have:
\begin{align*}
0 = \braket{b,b|b',b'} &= \frac{1}{N\sqrt{p(b)p(b')}}\sum_{a,a'}\braket{a,b,e_{a,b}|a',b',e_{a',b'}} = 0.
\end{align*}
Furthermore, we have:
\begin{align*}
1 = \braket{b,b|b,b} &= \frac{1}{N\cdot p(b)}\sum_{a,a'}\braket{a,b,e_{a,b}|a',b,e_{a',b}}\\
&=\frac{1}{N\cdot p(b)}\sum_{a}\bk{e_{a,b}} = 1.
\end{align*}
Thus, $\rw$ is an isometry and may be extended, using standard techniques, to a unitary operator implying it is an operation that $E$ may do within the laws of quantum physics.  We claim that $U = (I_{A_1}\otimes U_R)\rw$ is the desired attack operator satisfying the theorem statement.

\begin{figure}
  \centering
  \includegraphics[width=250pt]{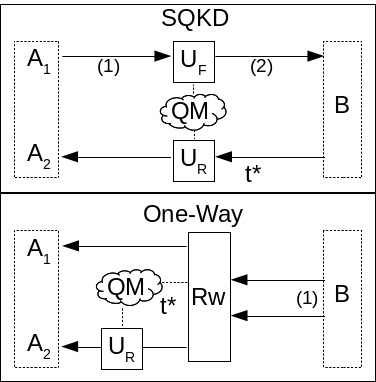}
  \caption{Showing the reduction from the semi-quantum protocol ($\psqkd$ and $\pe$, top) to the fully-quantum one-way protocol ($\ps$ bottom).  For the SQKD protocol, $A$ prepares qubits at time (1), Eve attacks, and then $B$ performs an operation $\MR$ or $\R$.  Time $t^*$ is after $B$'s operation.  On the other hand, for the fully-quantum protocol, $B$ prepares two qubits and sends both to $A$.  $E$ attacks with a specially designed $\rw$ operator resulting in a state at time $t^*$.  We claim a suitable $\rw$ operator can be constructed so that the density operators in both cases at time $t^*$ are identical.  Later, when proving general security of the one-way protocol, we do not require any special attack; clearly security of the SQKD protocol, then, would follow.  QM stands for $E$'s quantum memory.}\label{fig:proof}
\end{figure}

Refer to Figure \ref{fig:proof}.  Consider the case when $B$ chooses $\MR$.  At time $t^*$ (after $E$ attacks with $U_F$ and $B$'s $\MR$ operation, but \emph{before} $E$ attacks a second time with $U_R$), the joint state held by $A$, $B$, and $E$ using protocol $\pe$ is found to be:
\begin{equation}
\ket{\psi^{ent}_{MR}} = \frac{1}{\sqrt{N}}\sum_a\ket{a}_A\sum_b\ket{b, e_{a,b}, b}_{TEB}.
\end{equation}
Now, again, referring to Figure \ref{fig:proof}, consider the same case (namely, $B$ choosing $\MR$) but with the $\ps$ protocol.  In this event, $B$ prepares the state $\sum_b\sqrt{p(b)}\ket{b,b,b}_{A_1A_2B}$ and $E$ attacks with $\rw$.  The joint system then, at time $t^*$ is:
\begin{equation}
\ket{\phi_{MR}} = \sum_b\sqrt{p(b)}\left(\frac{\sum_a\ket{a,b,e_{a,b}}}{\sqrt{N\cdot p(b)}}\right)\otimes\ket{b}_B = \frac{1}{\sqrt{N}}\sum_a\ket{a}_{A_1}\sum_b\ket{b,e_{a,b},b}_{A_2EB} = \ket{\psi^{ent}_{MR}}
\end{equation}
Thus, after applying $\rw$, the state of the joint system for the case of $\ps$ is identical to that of $\pe$.  Of course, after applying $U_R$ (which happens in both scenarios since we constructed $U = (I_{A_1}\otimes U_R)\rw$), the systems will remain the same.  \emph{Thus, any measurement outcomes or entropy computations will be identical in both scenarios.}  It is trivial to show the same holds true in the $\R$ case for both protocols (in that case, the additional $\ket{b}$ term is no longer there but the algebra remains the same otherwise).  Thus, if one were to write out a density operator description of both protocols, tracing their evolution, they would be identical as the underlying systems are identical in all cases.  Note that the only thing $E$ could not ``rewind'' with $\rw$ is the probability distribution of $B$'s measurements (since he is now preparing).  Thus it is required that $E$ gets to choose the distribution $p(b)$ so that the probability distribution in $\ps$ matches that observed in $\pe$.  This completes the proof.
\end{proof}

Theorem \ref{thm:security1} implies that it is sufficient to prove security of the one-way protocol $\ps$.  Since any attack against $\pe$ can also be transformed into an attack against $\ps$, if we analyze a general attack against the latter, this automatically gives security against the former.  Indeed, there may be more attack strategies for $E$ against $\ps$ as $E$ has access to both $n$ qubit registers simultaneously; despite this, it is easier to analyze as it is a one-way protocol.  Furthermore, note that no party can distinguish between the two scenarios and, as a consequence, observed channel noise in the ``real'' SQKD protocol $\pe$ translate directly to observed statistics in the one-way protocol $\ps$.  Our goal is to prove security of $\pe$ (which proves security of $\psqkd$) and, given observed noise statistics there, if we prove security of $\ps$ given those same statistics, the key-rate can only be better in $\pe$ (since $\ps$ has potentially more attack strategies as mentioned).

\subsection{Proof of Security for $\ps$}
We now prove security of $\ps$.  In the following, we define $N = 2^n$ where $n$ is the user-defined number of qubits sent per iteration of our protocol.  Our proof of security is in three steps.  First, we compute the conditional entropy $H(A|E)$ in the case where $B$ chooses $\R$.  This, of course, is useless for key distillation as $B$ is completely independent of the state in this case, but it will be used later to argue about the entropy in the actual key-distillation state (i.e., when $B$ chooses $\MR$).  Second, we argue that $E$'s optimal attack must take on a particular form if $A$ and $B$ use the $\Z$ or $\X$ basis.  Third, and finally, we use these results, along with Winter's continuity bound on conditional entropy \cite{winter-cont}, to compute the entropy of $A$'s register conditioned on $E$'s quantum memory in the actual key-distillation state when $B$ chooses $\MR$ giving us the desired key-rate.

First, we need a channel scenario for the real $\pe$ protocol (which translates, as discussed, to observations for $\ps$).  Keeping in line with other high-dimensional QKD analyses \cite{QKD-high-dim-dep,QKD-QW}, we consider a symmetric attack modeled as a depolarization channel (\emph{which may even be enforced by users}):
\begin{equation}
\mathcal{E}_Q(\sigma) = \left(1-\frac{N}{N-1}Q\right)\sigma + \frac{Q}{N-1}I.
\end{equation}
where $\sigma$ is any $N$ dimensional quantum state.

We will assume the noise in the forward channel and reverse channel are the same and parameterized by $Q$ (though our analysis follows even if they are different, though the algebra complexity increases).  In the ``reflect'' case, we will use a depolarization parameter $Q_F$ - this captures the practical case that, for certain fiber channels, reflecting a quantum state back can ``undue'' some noise (but in the $\MR$ case this cannot happen as the ``measurement'' breaks any entanglement in the channel) \cite{lm05-ent,lucamarini2014quantum}.

Let $p(x|y)$ be the probability that a party observes $x$ given the sender sent $y$ (in the $\Z$ or $\X$ basis) in either the forward or reverse channel.  From this model, we have:
\begin{equation}\label{eq:pxy}
p(x|y) = \left\{\begin{array}{cl}
1-Q & \text{ if x = y}\\
\frac{Q}{N-1} & \text{ otherwise}
\end{array}\right.
\end{equation}
In $\pe$, the probability that $A_1^Z$ (i.e., after measuring) is $a$, for any particular $a$, is simply $p(a) = 1/N$.  Furthermore the probability that $B$ measures $b$ is $\sum_ap(b|a)p(a) = \frac{1}{N}( 1-Q + (N-1)\frac{Q}{N-1}) = 1/N$.  Thus, we set $p(b) = 1/N$ when analyzing $\ps$ ($E$'s choice here must conform to the observed statistics in the ``real'' protocol $\pe$).

Let $p(a,b,c)$ be the probability that, in the case of $\MR$, if all parties measure in the $\Z$ basis, $A_1$ measures $a$, $B$ measures $b$, and $A_2$ measures $c$ (recall $A_1$ is $A$'s first $n$-qubit register and $A_2$ is her second register).  Then, since this is a classical probability distribution, by the chain rule it holds that:
\begin{equation}
p(a,b,c) = p(c|b,a)\cdot p(b|a)\cdot p(a).
\end{equation}
We will assume that \emph{in the $\MR$ case} of $\pe$, the two channels act independently and, so, $p(c|b,a) = p(c|b)$.  That is, $A$'s measurement in the return channel, depends only on what $B$ actually sends.  This is a very realistic noise scenario \emph{and can even be enforced by the users} - $A$ and $B$ will simply abort if they do not observe this (natural) behavior.  Of course, as discussed, we do not assume the two channels act independently if $B$ chooses to $\R$ (such an assumption would not be natural nor could it be enforced and so we do not make it here).  Under these assumptions, it is not difficult to see that:
\begin{equation}\label{eq:pabc}
p(a,b,c) = \frac{1}{N} \times \left\{\begin{array}{cl}
\beta^2 & \text{ if $c = b$ and $b = a$}\\
\alpha\beta & \text{ if $c \ne b$ and $b = a$}\\
\alpha\beta & \text{ if $c = b$ and $b \ne a$}\\
\alpha^2 & \text{ if $c \ne b$ and $b \ne a$}
\end{array}\right.
\end{equation}
where:
\begin{align}
\alpha = \frac{Q}{N-1} && \beta = 1-Q.
\end{align}

Given these observed channel statistics in $\pe$ we now turn to $\ps$ using this same distribution on measurement events.  Ultimately, our goal is to compute a lower-bound on the key-rate: $H(A_1^Z|E)_\mu - H(A_1^Z|B_1^Z)_\mu$, where $\mu_{A_1B_1E}$ is the density operator describing an iteration of the protocol in the $\MR$ case.

$ $\newline\newline\textbf{First Step - Entropy in the Reflect Case:}
Let $\rho_{A_1A_2BE}$ be the density operator describing the state of the system (\emph{before measurements are made by any party}) if $B$ chooses $\R$.  Similarly, let $\mu_{A_1A_2BE}$ be the density operator in the case $B$ chooses $\MR$.  Since key-bits are only distilled in this $\MR$ case, to compute the key-rate of the protocol we will require a bound on the von Neumann entropy $H(A_1^Z|E)_\mu$.  However, we will actually, first, bound $H(A_1^Z|E)_\rho$ and later argue that, due to continuity of entropy \cite{winter-cont}, the difference in entropy between the two systems, $\rho$ and $\mu$, cannot be ``too large.''

Consider the state $\rho_{A_1A_2BE}$.  In this $\R$ case, $B$'s system is completely independent of all other systems; thus $\rho_{A_1A_2BE} \equiv \rho_{A_1A_2E}\otimes\kb{0}_B$ and so the $B$ portion does not factor into any entropy equations and may be ignored.  Using the entropic uncertainty relation proven in \cite{QKD-uncertainty} (see Equation \ref{eq:uncertainty}), we know:
\[
H(A_1^Z|E)_\rho \ge n - H(A_1^F|A_2)_\rho \ge n - H(A_1^F|A_2^F)_\rho,
\]
where the second inequality follows from the fact that measurements cannot decrease uncertainty.  If we could distill a key from $\rho$, we would be finished - in fact, the above would be the case when $A_1$ is attempting to distill a key with herself, ``$A_2$'' which, of course, is meaningless from a practical standpoint.  However, as we now show, knowing the entropy in $\rho$ allows us to bound the entropy in $\mu$ (which is what we actually want in order to compute the key-rate of our protocol).

Consider $E$'s attack operator $U$ against $\ps$. Without loss of generality, we may write $U$'s action on basis states of the form $\ket{b,b}$ as follows:
\begin{equation}
U\ket{b,b}\otimes\ket{\chi}_E = \sum_{a=0}^{N-1}\sum_{c=0}^{N-1}\ket{a,c, e_{a,b,c}},
\end{equation}
where the $\ket{e_{a,b,c}}$ are arbitrary states in $E$'s ancilla (again, we assume without loss of generality that $E$'s ancilla is cleared to some initial pure state $\ket{\chi}_E$).  Note that we are not assuming a particular structure to this attack (e.g., we do not assume it consists of the $\rw$ operator used in the proof of Theorem \ref{thm:security1} - instead, it may be arbitrary and if we prove security here, we will gain security of $\pe$ since it will cover any attack against that protocol).

Consider $\rho_{A_1^ZE}$, i.e., the state of the system after $A$ measures the $A_1$ register in the $\Z$ basis and tracing out $A_2$.  Tracing the evolution of the state in this case, and recalling that $p(b) = 1/N$ due to our (enforceable) symmetry assumption, we find:
\begin{equation}\label{eq:thm2-rho}
\rho_{A_1^ZE} = \frac{1}{N}\sum_{a=0}^{N-1}\kb{a}\otimes \left(\sum_{c=0}^{N-1}P\left[\sum_{b=0}^{N-1}\ket{e_{a,b,c}}\right]\right),
\end{equation}
where $P(z) = zz^*$.

On the other hand, tracing the evolution of the protocol in the case $B$ chooses $\MR$, gives us the following operator:
\begin{equation}\label{eq:thm2-mu}
\mu_{A_1^ZE} = \frac{1}{N}\sum_{a=0}^{N-1}\kb{a}\otimes \left(\sum_{c=0}^{N-1}\sum_{b=0}^{N-1}\kb{e_{a,b,c}}\right).
\end{equation}

Our goal in the remainder of the security proof is to bound the difference between $H(A_1^Z|E)_\rho$ and $H(A_1^Z|E)_\mu$.  To do so, we will use Winter's continuity bound \cite{winter-cont} and in particular, the case derived for classical-quantum states.  This bound states that (rewriting in terms of our notation of course):
\begin{equation}\label{eq:cont-bound}
|H(A_1^Z|E)_\rho - H(A_1^Z|E)_\mu| \le \Delta\log A_1^Z + (1+\Delta)H\left(\frac{\Delta}{1+\Delta}\right),
\end{equation}
where:
\[
\Delta = \frac{1}{2}\tdl \rho_{A_1^ZE} - \mu_{A_1^ZE}\tdr.
\]
Of course $\log A_1^Z = n$.  Thus, our goal is to determine an upper-bound on the trace distance $\Delta$.  Note that an upper-bound will only increase the distance between the two entropies causing the key-rate to drop.  Thus by finding an upper-bound, we determine a worst-case key-rate and the actual key-rate can only be higher.

By elementary properties of trace distance, along with the triangle inequality, we have:
\begin{equation}\label{eq:delta-ac}
\Delta \le \frac{1}{2N}\sum_{a,c=0}^{N-1}\underbrace{\tdl P\left(\sum_{b=0}^{N-1}\ket{e_{a,b,c}}\right) - \sum_{b=0}^{N-1}\kb{e_{a,b,c}} \tdr}_{\Delta_{a,c}} = \frac{1}{2N}\sum_{a,c}\Delta_{a,c}.
\end{equation}

$ $\newline\newline\textbf{Second Step - Structure of $E$'s Attack Operator:}
Before computing $\Delta$, we argue now that $E$'s optimal attack operator has a particular structure to it.
As discussed earlier, let $p(a,b,c)$ denote the probability that measuring $A_1$ results in $a$; measuring $B$ results in $b$; and measuring $A_2$ results in $c$ (where these measurements are performed in the $\Z$ basis in the $\MR$ case; thus $a,b,c \in \{0,1,\cdots,N-1\}$).  It is not difficult to see that $p(a,b,c) = \bk{e_{a,b,c}}/N$.  Indeed, note that the state $\mu_{A_1^ZA_2^ZBE}$ (i.e., the case where $B$ chooses $\MR$, but before tracing out $A_2^Z$ and $B$ which we did for Equation \ref{eq:thm2-mu}) is found to be:
\[
\mu_{A^Z_1A^Z_2BE} = \frac{1}{N}\sum_a\kb{a}_{A_1}\otimes\sum_c\kb{c}_{A_2}\otimes\sum_b\kb{b}_B\otimes\kb{e_{a,b,c}},
\]
from which it is clear that $p(a,b,c) = \bk{e_{a,b,c}}/N$.  Since $N$ is known and since $p(a,b,c)$ is a value that can be observed by the parties running the protocol, this implies $\bk{e_{a,b,c}}$ is also an observable quantity.

We now claim that it is to $E$'s advantage to choose her attack such that for any \emph{fixed} $a, c$, it holds that:
\begin{equation}\label{eq:state-orth}
\braket{e_{a,b,c}|e_{a,b',c}} = \left\{\begin{array}{cl}
N\cdot p(a,b,c) & \text{ if } b = b'\\
0&\text{ if } b \ne b'\end{array}\right.
\end{equation}
Indeed, orthogonal ancilla states cannot increase her uncertainty, thus the only reason to make these states non orthogonal would be if, by doing so, she could make some other, potentially ``more important'' vectors closer to orthogonal (e.g., the non-error cases such as $\braket{e_{0,0,0}|e_{1,1,1}}$) while still falling within the observed noise statistics.  But the inner-product $\braket{e_{a,b,c}|e_{a,b',c}}$ does not contribute to the observed noise in any way, assuming basis $\X$ is used, and thus she might as well set them to be orthogonal potentially decreasing her overall uncertainty (but certainly not increasing it).

Clearly the inner-product $\braket{e_{a,b,c}|e_{a,b',c}}$ does not contribute to the $\Z$ basis noise when $b \ne b'$.  We thus consider the $\X$ basis noise.  Consider the case when $B$ chooses $\R$ in which case the state arriving to $A$, before measuring, is:
\begin{equation}\label{eq:thm:ref-state}
\frac{1}{2^{n/2}}\sum_{a,c}\ket{a,c}\ket{g_{a,c}},
\end{equation}
where $\ket{g_{a,c}} = \sum_b\ket{e_{a,b,c}}$.  Since the above is normalized, it holds that:
\begin{equation}\label{eq:gab-norm}
\frac{1}{2^n}\sum_{a,c}\bk{g_{a,c}} = 1.
\end{equation}
Now, changing basis, we may write $\ket{j} = \sum_x\beta_{x,j}\ket{F_x}$, where $\beta_{x,j} = \braket{F_x|j}$.  Clearly, due to our choice of basis $\X$, it holds that $|\beta_{x,j}|^2 = 1/2^n$.  Taking Equation \ref{eq:thm:ref-state} and changing basis in both the $A_1$ and $A_2$ registers yields:
\[
\frac{1}{2^{n/2}}\sum_{x,y}\ket{F_x, F_y}\left(\sum_{a,c}\beta_{x,a}\beta_{y,c}\ket{g_{a,c}}\right).
\]
Thus, the probability that $A_1$ measures $F_x$ and $A_2$ measures $F_y$, for any $x,y$ is:
\begin{align*}
& \frac{1}{2^n}\left|\sum_{a,c}\beta_{x,a}\beta_{y,c}\ket{g_{a,c}}\right|^2\\
=& \frac{1}{2^n}\sum_{a,c}\frac{1}{2^{2n}}\bk{g_{a,c}} + \sum_{(a,c)\ne(a',c')}\beta_{x,a}\beta_{x,a'}^*\beta_{y,c}\beta_{y,c'}^*\braket{g_{a,c}|g_{a',c'}}\\
=&\frac{1}{2^{2n}} + \sum_{(a,c)\ne(a',c')}\beta_{x,a}\beta_{x,a'}^*\beta_{y,c}\beta_{y,c'}^*\braket{g_{a,c}|g_{a',c'}},
\end{align*}
where for the third equality, we use Equation \ref{eq:gab-norm}.
Note that $\braket{g_{a,c}|g_{a',c'}}$, for $(a,c) \ne (a',c')$ has no terms of the form $\braket{e_{a,b,c}|e_{a,b',c}}$ (since either $a'$ or $c'$ will not equal $a$ or $c$).  Thus the $\braket{e_{a,b,c}|e_{a,b',c}}$ inner product cannot affect any observed noise statistic.  Therefore there is no advantage to $E$ in making it non-orthogonal as it cannot benefit her by ``hiding'' other states in the noise of the channel (e.g., she cannot use $\braket{e_{a,b,c}|e_{a,b',c}}$ to increase the orthogonality of other vectors to her advantage while still keeping within the observed noise statistics).  We may therefore assume the attack operator $U$ is such that Equation \ref{eq:state-orth} applies.  Note that this proof would not hold if $|\beta_{i,j}|^2 \ne 1/2^n$ for all $i,j$.

Thus, for any fixed $a$ and $c$, we may define an orthonormal basis $\{\ket{\nu_b^{(a,c)}}\}_{b=0}^{N-1}$ and write:
\[
\ket{e_{a,b,c}} = \sqrt{N\cdot p(a,b,c)}\ket{\nu_b^{(a,c)}}.
\]
Note that we do not assume any relation between these vectors for differing $a$ and $c$.  I.e., we do not make any assumptions on the value $\braket{\nu_b^{(a,c)}|\nu_{b'}^{(a',c')}}$ when $a\ne a'$ or $c \ne c'$.

$ $\newline\newline\textbf{Third Step - Continuity Bound Analysis:}
From the above analysis on the structure of $E$'s optimal attack operator, we may write $\Delta_{a,c}$, defined in Equation \ref{eq:delta-ac}, as:
\begin{align}
\Delta_{a,c} &= \tdl P\left(\sum_{b=0}^{N-1}\sqrt{N\cdot p(a,b,c)}\ket{\nu_b^{(a,c)}}\right) - \sum_{b=0}^{N-1}N\cdot p(a,b,c)\kb{\nu_b^{(a,c)}} \tdr\notag\\
&=N \tdl P\left(\sum_{b=0}^{N-1}\sqrt{p(a,b,c)}\ket{b}\right) - \sum_{b=0}^{N-1}p(a,b,c)\kb{b} \tdr\label{eq:thm2-delta-ac2}
\end{align}
where the last equality follows from the fact that trace distance is invariant to changes in basis and, again, we use $P(z) = zz^*$.


Recall our description of the channel, and in particular the value of $p(a,b,c)$ given in Equation \ref{eq:pabc}.  Note that, if $Q = 0$, then it is easy to see that $\Delta_{a,c} = 0$ for all $a,c$ and so we are done.  Thus, in the following, we will consider $0 < Q < 1/2$.  Due to the symmetry in a depolarization channel as clearly seen in the expression for $p(a,b,c)$ in Equation \ref{eq:pabc} (\emph{again, this may even be enforced by users)}, there are two cases to consider, first when $c = a$ and second when $c \ne a$.  For the first, we have:
\begin{align}
\Delta_{a,a} &= N\tdl P\left(\sum_{b=0}^{N-1}\sqrt{p(a,b,a)}\ket{b}\right) - \sum_{b=0}^{N-1}p(a,b,a)\kb{b}\tdr\notag\\
&=N\tdl\sum_{b\ne b'}\sqrt{p(a,b,a)p(a,b',a)}\ket{b}\bra{b'}\tdr.
\end{align}

Let $X$ be the operator $X = N\sum_{b\ne b'}\sqrt{p(a,b,a)\cdot p(a,b'a)}\ket{b}\bra{b'}$.  Thus $\Delta_{a,a} = ||X||$.  Since it is Hermitian, we may decompose $X$ as:
\begin{equation}\label{eq:A-expansion}
X = \sum_{j=0}^{N-1}\lambda_j\kb{v_j},
\end{equation}
where $\{\ket{v_j}\}$ are orthogonal eigenvectors and $\lambda_j$ are (real) eigenvalues; thus $X\ket{v_j} = \lambda_j\ket{v_j}$ for all $j=0,\cdots, N-1$ and, of course, $||X|| = \sum_j|\lambda_j|$.  Consider a particular eigenvector $\ket{v} = \ket{v_j} = \sum_ix_i\ket{i}$.  Then:
\begin{equation}
X\ket{v} = N\sum_{b=0}^{N-1}\underbrace{\left(\sum_{\substack{i=0\\i\ne b}}^{N-1}x_i\sqrt{p(a,b,a)\cdot p(a,i,a)}\right)}_{y_b}\ket{b} = N\sum_by_b\ket{b}.
\end{equation}
Thus, for $\lambda = \lambda_j$ to be the corresponding eigenvalue, it must hold that $Ny_b = \lambda x_b$ for all $b=0,\cdots,N-1$.  Note that, when $b=a$, it holds that:
\begin{align}
&Ny_a = \lambda x_a \iff N\sum_{\substack{i=0\\i\ne a}}^{N-1}x_i\sqrt{p(a,a,a)\cdot p(a,i,a)}=\lambda x_a\iff\alpha\beta\sum_{i\ne a}x_i = \lambda x_a.\label{eq:thm-main-A-mat1}
\end{align}
When $b \ne a$, then $Ny_b$ simplifies to:
\begin{align}
N\sum_{i\ne b}x_i\sqrt{p(a,b,a)\cdot p(a,i,a)} = N\left(x_a\sqrt{p(a,b,a)\cdot p(a,a,a)} + \sum_{\substack{i\ne b\\i\ne a}}x_i\sqrt{p(a,b,a)\cdot p(a,i,a)}\right)\notag
\end{align}
and thus it must hold that:
\begin{equation}\label{eq:thm-main-A-mat2}
\alpha\beta x_a + \alpha^2\sum_{\substack{i\ne b\\i \ne a}}x_i = \lambda x_b.
\end{equation}

Now, assume that there exists a $k\ne k'$ such that $x_k \ne x_{k'}$ and both $k$ and $k'$ are not equal to $a$ (we will handle the case when this is not true afterwards).  From Equation \ref{eq:thm-main-A-mat2} we have, using the case when $b = k$ and $b = k'$ respectively:
\begin{align*}
\alpha\beta x_a + \alpha^2\sum_{\substack{i\ne k\\i \ne a}}x_i &= \lambda x_k\\
\alpha\beta x_a + \alpha^2\sum_{\substack{i\ne k'\\i \ne a}}x_i &= \lambda x_{k'}.
\end{align*}
Subtracting these two expressions yields:
\begin{align}
&\alpha^2(x_{k'} - x_k) = \lambda (x_k - x_{k'})\notag\\
\Rightarrow &\lambda  = -\alpha^2.
\end{align}
We next claim the geometric multiplicity of this eigenvalue is $N-2$ and, thus, this eigenvalue appears $N-2$ times in Equation \ref{eq:A-expansion}.  Consider the operator $X-\lambda I$.  By choosing a suitable basis we may write this in matrix form as:
\begin{equation}
X-\lambda I = \left(
\begin{array}{ccccc}
-\lambda & \alpha\beta & \alpha\beta & \cdots & \alpha\beta\\
\alpha\beta & -\lambda & \alpha^2  & \cdots & \alpha^2\\
\alpha\beta & \alpha^2 & -\lambda & \cdots & \alpha^2\\
\vdots & \vdots & \vdots&\ddots&\vdots\\
\alpha\beta & \alpha^2 & \alpha^2 & \cdots & -\lambda
\end{array}\right)
\end{equation}
Substituting $\lambda = -\alpha^2$ it is clear that the rank of $X - (-\alpha^2)I$ is at most two.  Thus the geometric multiplicity is at least $N-2$ (and indeed is exactly $N-2$ except when $Q=0$ or $Q = 1-1/N$; but the first case is considered separately as mentioned, and the second case implies $Q > 1/2$ which is much larger than our evaluations later and so not considered).  Therefore, exactly $N-2$ of the eigenvalues of $X$ are $-\alpha^2 = -\frac{Q^2}{(N-1)^2}$.

The remaining two eigenvalues are found when there does not exist $k \ne k'$ (where $k \ne a$ and $k' \ne a$) such that $x_k \ne x_{k'}$.  In this case we have $x_k = x_{k'} = x$ for all $k, k'$ not equal to $a$.  Using Equation \ref{eq:thm-main-A-mat1} we find:
\[
\alpha\beta(N-1)x = \lambda x_a \Longrightarrow x_a = \frac{\alpha\beta(N-1)x}{\lambda}.
\]
Note that the above equation forces $x \ne 0$ as, otherwise, $x_a$ is also $0$ and so $\ket{v}$ would be the zero vector and not an eigenvector of Hermitian operator $X$.  Substituting this into Equation \ref{eq:thm-main-A-mat2} (for any $b \ne a$) we find:
\begin{align*}
&\alpha\beta\left(\frac{\alpha\beta(N-1)x}{\lambda}\right) + \alpha^2(N-2)x = \lambda x\\
\iff & \lambda^2 - \alpha^2(N-2)\lambda - \alpha^2\beta^2(N-1) = 0,
\end{align*}
thus leading us to the two remaining eigenvalues, which we denote $\lambda^X_\pm$:
\[
\lambda^X_\pm = \frac{1}{2}\left(\alpha^2(N-2)\pm \alpha\sqrt{\alpha^2(N-2)^2 + 4\beta^2(N-1)}\right).
\]
Since there was no dependence on $a$ in the above analysis, this leads us to conclude that:
\begin{equation}
\Delta_{a,a} = (N-2)\frac{Q^2}{(N-1)^2} + |\lambda^X_+| + |\lambda^X_-|.
\end{equation}

We next consider the case when $c \ne a$ and compute $\Delta_{a,c}$.  Following the same logic as before, fix a particular $c \ne a$ and consider the operator $Y = N\sum_{b\ne b'}\sqrt{p(a,b,c)\cdot p(a,b',c)}\ket{b}\bra{b'}$ (and so $\Delta_{a,c} = ||Y||$).  Let $\ket{v} = \sum_ix_i\ket{i}$ be an eigenvector of $Y$ such that $Y\ket{v} = \lambda\ket{v}$.  Then:
\[
Y\ket{v} = N\sum_{b=0}^{N-1}\underbrace{\left(\sum_{\substack{i=0\\i \ne b}}^{N-1}x_i\sqrt{p(a,b,c)\cdot p(a,i,c)}\right)}_{z_b}\ket{b} = N\sum_bz_b\ket{b}.
\]
To satisfy the equation $Y\ket{v} = \lambda\ket{v}$ we require $Nz_b = \lambda x_b$ for all $b=0,\cdots, N-1$.  There are three cases of $b$ to consider here: $b=a$, $b=c$ and $b\ne a,c$.  For each of these cases we find:
\begin{align}
b = a: Nz_a = \lambda x_a \iff && \alpha\beta x_c + \sum_{\substack{i\ne c\\ i \ne a}}x_i\sqrt{\alpha^3\beta} = \lambda x_a\label{eq:thm-main-X1}\\
b = c: Nz_c = \lambda x_c \iff && \alpha\beta x_a + \sum_{\substack{i \ne c \\ i \ne a}}x_i\sqrt{\alpha^3\beta} = \lambda x_c\label{eq:thm-main-X2}\\
b\ne a,c: Nz_b = \lambda x_b \iff && \sqrt{\alpha^3\beta}x_a + \sqrt{\alpha^3\beta}x_c + \sum_{\substack{i\ne a\\i \ne b\\ i \ne c}}\alpha^2x_i = \lambda x_b.\label{eq:thm-main-X3}
\end{align}

As with the previous operator $X$, we break this up into several cases depending on the eigenvector $\ket{v}$.  For the first case, assume there exists $k \ne k'$ with $k \ne a,c$ and $k'\ne a,c$ such that $x_k \ne x_{k'}$.  Then, using Equation \ref{eq:thm-main-X3}, for $b = k$ and $b=k'$ and subtracting the resulting expressions yields:
\begin{equation}
\alpha^2(x_{k'} - x_k) = \lambda(x_k - x_{k'}) \Longrightarrow \lambda = -\alpha^2
\end{equation}
We claim this eigenvalue has geometric multiplicity $N-3$.  Consider the operator $Y - \lambda I$ and, as before, by considering a suitable basis, we may write this in matrix form as:
\begin{equation}
Y-\lambda I = \left(
\begin{array}{cccccc}
-\lambda & \alpha\beta & \sqrt{\alpha^3\beta} & \sqrt{\alpha^3\beta} & \cdots & \sqrt{\alpha^3\beta}\\
\alpha\beta & -\lambda & \sqrt{\alpha^3\beta} & \sqrt{\alpha^3\beta} & \cdots & \sqrt{\alpha^3\beta}\\
\sqrt{\alpha^3\beta} & \sqrt{\alpha^3\beta} & -\lambda & \alpha^2 & \cdots & \alpha^2\\
\sqrt{\alpha^3\beta} & \sqrt{\alpha^3\beta} & \alpha^2 & -\lambda & \cdots & \alpha^2\\
\vdots & \vdots & \vdots & \cdots & \ddots & \vdots\\
\sqrt{\alpha^3\beta} & \sqrt{\alpha^3\beta} & \alpha^2 & \alpha^2 & \cdots & -\lambda
\end{array}\right)
\end{equation}
From this, it is evident that the rank of $Y - (-\alpha^2)I$ is three and so the geometric multiplicity of the eigenvalue $-\alpha^2$ is $N-3$ (again, assuming $Q \ne 0$ and $Q \ne 1-1/N$ which holds since $0 < Q < 1/2$.

Thus, there are 3 more eigenvalues.  Next, consider the case if $x_a \ne x_c$.  In this case, subtracting Equation \ref{eq:thm-main-X1} and \ref{eq:thm-main-X2} yields:
\begin{equation}
\alpha\beta(x_c-x_a) = \lambda(x_a-x_c) \Longrightarrow \lambda = -\alpha\beta
\end{equation}

Finally, consider the case where $x_a = x_c = x_1$ and $x_k = x_{k'} = x_2$ for every $k\ne k'$ and $k,k' \ne a,c$.  In this case, Equation \ref{eq:thm-main-X3} simplifies to:
\begin{equation}\label{eq:thm-main-B3-2}
2\sqrt{\alpha^3\beta}x_1 + \alpha^2(N-3)x_2 = \lambda x_2
\end{equation}
Note that this implies $x_2 \ne 0$ as, otherwise, $x_1 = 0$ and so $\ket{v}$ is the zero vector and not an eigenvector.

Equation \ref{eq:thm-main-X1} yields:
\[
\alpha\beta x_1 + (N-2)\sqrt{\alpha^3\beta}x_2 = \lambda x_1
\]
Note that the above equation also implies that $\lambda \ne \alpha\beta$ since if it were, we would have $x_2 = 0$ which, as already discussed, is not true.  Thus we may solve:
\[
x_1 = \frac{(N-2)x_2\sqrt{\alpha^3\beta}}{\lambda-\alpha\beta}.
\]
Substituting this into Equation \ref{eq:thm-main-B3-2} yields:
\begin{align}
&\frac{2x_2\alpha^3\beta(N-2)}{\lambda-\alpha\beta} + \alpha^2(N-3)x_2 = \lambda x_2\notag\\
\iff&2\alpha^3\beta(N-2) + \alpha^2(N-3)(\lambda-\alpha\beta) = \lambda(\lambda-\alpha\beta).
\end{align}
Solving the above quadratic for $\lambda$ gives us the two remaining eigenvalues which we denote $\lambda^Y_\pm$.  After some algebra, these eigenvalues are found to be:
\begin{equation}
\lambda^Y_\pm = \frac{1}{2}\left[\alpha\beta + \alpha^2(N-3) \pm \alpha\sqrt{\left(\beta+\alpha[N-3]\right)^2 + 4\alpha\beta(N-1)}\right]
\end{equation}
Since the above arguments were for arbitrary $a \ne c$, this gives us the following:
\begin{align}
\Delta_{a,c} &= (N-3)\alpha^2 + \alpha\beta + |\lambda^Y_+| + |\lambda^Y_-|\notag\\
&= (N-3)\frac{Q^2}{(N-1)^2} + \frac{Q(1-Q)}{N-1} + |\lambda^Y_+| + |\lambda^Y_-|.
\end{align}

Thus, we conclude:
\begin{align}
\Delta &= \frac{1}{2N}\sum_{a,c=0}^{N-1}\Delta_{a,c} = \frac{1}{2N}\sum_a\Delta_{a,a} + \frac{1}{2N}\sum_{a\ne c}\Delta_{a,c}\\
&= \frac{1}{2}\left((N-2)\frac{Q^2}{(N-1)^2} + |\lambda^X_+| + |\lambda^X_-| + (N-1)\left[(N-3)\frac{Q^2}{(N-1)^2} + \frac{Q(1-Q)}{N-1} + |\lambda_+^Y| + |\lambda_-^Y|\right]\right).
\end{align}
At first glance, this expression may seem to scale exponentially with $n$ (since $N = 2^n$).  However, note that $\lambda_\pm$ (for both the $X$ and $Y$ operators) are multiples of $\alpha$, which, itself, is a multiple of $1/(N-1)$.


Returning to $\ps$, we apply the Winter continuity bound (Equation \ref{eq:cont-bound}) to attain:
\begin{align}
H(A_1^Z|E)_\mu &\ge H(A_1^Z|E)_\rho - \Delta\log N - (1+\Delta)H\left(\frac{\Delta}{1+\Delta}\right)\notag\\
&\ge n - H(A_1^F|A_2^F)_\rho - \Delta\log N - (1+\Delta)H\left(\frac{\Delta}{1+\Delta}\right)\notag\\
&\ge n(1 - \Delta) - (1+\Delta)H\left(\frac{\Delta}{1+\Delta}\right) - H(A_1^F|A_2^F)_\rho.
\end{align}
To finish the key-rate computation, we need $H(A_1^F|A_2^F)_\rho$ and $H(A_1^Z|B^Z)_\mu$.  The first is determined through the observed values $p_{i,j}^F$, which we use to denote the probability that $A$ observes $\ket{F_i}$ (in $A_1$) and $\ket{F_j}$ (in $A_2$) conditioned on the event $B$ choose $\R$; the second is determined through the observed values $p_{i,j}^Z$ which we use to denote the probability that $B$ observes $\ket{j}$ and $A$ observes $\ket{i}$ in $A_1$ (we use $A_1$ as this is the register used for key-distillation) conditioned on the event $B$ choose $\MR$.  In both cases, $i,j \in \{0,1,\cdots,N-1\}$.  Clearly these are observable values allowing $A$ and $B$ to compute these final (classical) entropy expressions.  Since we are considering a symmetric attack modeled by the depolarization channel described in Equation \ref{eq:pxy}, we have $p_{i,j}^Z = \frac{1}{N}p(j|i)$ and so we compute the joint entropy as:
\begin{align*}
H(A_1^ZB^Z)_\mu &= -\sum_{i,j}p_{i,j}^Z \log p_{i,j}^Z = -\sum_i \frac{1}{N}p(i|i)\log_2 p(i|i) - \sum_{i\ne j} \frac{1}{N}p(j|i)\log_2p(j|i)\\
&=(1-Q)\log_2\frac{1-Q}{N} - Q\log\frac{Q}{N(N-1)}\\
&= n + Q\log_2(N-1) + H(Q).
\end{align*}
It is not difficult to show that $H(B^Z)_\mu = n$ (since the attack is symmetric, $B$'s probability of observing any particular value $\ket{j}$ is uniform).  Thus the conditional entropy is simply:
\[
H(A_1^Z|B^Z) = Q\log_2(N-1) + H(Q).
\]

The case for the $\X$ basis is identical, though we use a different noise parameter $Q_F$ to parameterize the channel in this case (since the noise may be different in the reflection case as discussed earlier).  In this case we have:
\[
H(A_1^F|A_2^F) = Q_F\log_2(N-1) + H(Q_F).
\]

Our final key-rate expression, therefore is:
\begin{align}
r &=H(A_1^Z|E) - H(A_1^Z|B^Z)\notag\\
&\ge n(1 - \Delta) - (1+\Delta)H\left(\frac{\Delta}{1+\Delta}\right) - (Q+Q_F)\log_2\left(2^n-1\right) - H(Q) - H(Q_F).\label{eq:keyrate-final}
\end{align}

Note that the above assumed collective attacks.  Ordinarily, one may extend such computations done for the collective attack case to prove security against arbitrary, general, attacks by using de Finetti style arguments or post-selection techniques \cite{de-finetti1,postselection1}.  We suspect that this result holds for our protocol, however we leave a complete proof of that for future work.

\subsection{Evaluation}
We evaluate our key-rate bound, Equation \ref{eq:keyrate-final}, in two scenarios.  First, we assume in the \emph{reflection} case, that the reverse channel is \emph{independent} of the forward and, so, $Q_F = 2Q(1-Q)$ shown in Figure \ref{fig:keyrate-ind}.  In the second \emph{dependent} case, we assume $Q_F = Q$ shown in Figure \ref{fig:keyrate-dep}.  We note that, similar to the fully-quantum case \cite{QKD-high-dim-dep,QKD-high-dim-50}, as the dimension increases, the noise tolerance also surpasses the single qubit case.  Thus, we prove that this high-dimensional advantage, known for fully-quantum protocols, also applies to the semi-quantum model.  We also observe numerically that, as $n$ increases, the maximal noise tolerance tends to approach $26\%$ in the independent case and $30\%$ in the dependent case.  As mentioned, fully-quantum high dimensional QKD protocols can tolerate up to $50\%$ error as the dimension increases; thus, while not as high as the fully-quantum case (which, perhaps, is to be expected), it is higher than any other semi-quantum protocol to-date.  Indeed, the highest known semi-quantum protocol \cite{SQKD-high-noise} can tolerate up to $17.8\%$ in the independent case (as opposed to $26\%$ here) and $26\%$ in the dependent case (as opposed to $30\%$ here).  Of course, our Equation \ref{eq:keyrate-final} is only a lower-bound - future work may improve this.  In particular, the use of mismatched measurements (needed to attain a high noise tolerance in \cite{SQKD-high-noise}) may greatly benefit our analysis here.  This we leave as an interesting future research direction.

\begin{figure}
  \centering
  \includegraphics[width=250pt]{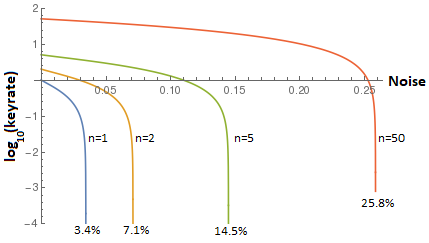}
  \caption{Key-rate of our high-dimensional SQKD protocol when $Q_F = 2Q(1-Q)$.  Here we plot the case for $n=1, 2, 5,$ and $50$.}\label{fig:keyrate-ind}
\end{figure}

\begin{figure}
  \centering
  \includegraphics[width=250pt]{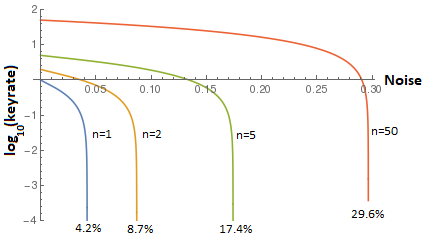}
  \caption{Key-rate of our high-dimensional SQKD protocol when $Q_F = Q$.  Here we plot the case for $n=1, 2, 5,$ and 50.}\label{fig:keyrate-dep}
\end{figure}

\section{Closing Remarks}

In this paper, we designed a new high-dimensional semi-quantum key distribution protocol and performed an information theoretic security analysis.  To conduct this security analysis, we developed several new techniques for high-dimensional protocols over two-way quantum channels which may be applicable to other (S)QKD protocols.  In particular we showed how one may reduce a two-way, high dimensional, semi-quantum protocol to a one-way protocol which is easier to analyze.   Thus, we produced new security results of broad application.  We also proved that high-dimensional quantum systems can benefit communication in the semi-quantum model just as they do in fully-quantum key distribution.

Many interesting future problems remain open.  For one thing, it would be interesting to see if our proof technique can be applied to the high-dimensional quantum-walk based SQKD protocol introduced in \cite{QKD-QW}.  If so, we would then be able to compare noise tolerance properties of the two protocols.  It would also be interesting to see if we can improve our bound and technique here.  One factor contributing to a potentially lower key-rate bound is our use of a continuity bound.  Other methods may produce more optimistic results.


\end{document}